\documentclass[pmlr,twocolumn,a4paper,10pt]{jmlr}
\usepackage{multirow}
\usepackage{bm}
\usepackage{amsmath, amssymb, amsfonts}
\usepackage{array}
\usepackage{enumerate}
\usepackage{bbm}
\usepackage{color}
\usepackage{float}
\usepackage[mathscr]{euscript}
\usepackage{mathtools}
\usepackage{graphicx,subcaption}
\usepackage{framed}
\usepackage{setspace}
\usepackage{mdframed}

\usepackage{bussproofs}
\usepackage{proof}

\newcommand {\sequent}{\rhd}

\newcommand {\closure}{\mathsf{Cn}}

\usepackage{thmtools}
\usepackage{environ}
\usepackage{mathptmx}
\usepackage[draft]{todonotes} 
  \usepackage{times}
\usepackage{pgfplots}
\usetikzlibrary{matrix}
\usepackage[toc,page]{appendix}

  \usetikzlibrary{positioning,calc,fit,shapes.geometric,patterns}
  \pgfdeclarelayer{background}
  \pgfdeclarelayer{foreground}
  \pgfsetlayers{background,main,foreground}
  \tikzstyle{vec}=[circle,inner sep=1pt,outer sep=-1pt,fill]
  \tikzstyle{border}=[thick]
  \tikzstyle{favborder}=[border,dotted]
  \tikzstyle{exclborder}=[border,dashed]
\usepackage{etex,etoolbox}
\usepackage{hyperref}
 \usepackage{calrsfs}

  \usetikzlibrary{positioning,calc,fit,shapes.geometric,patterns}
  \pgfdeclarelayer{background}
  \pgfdeclarelayer{foreground}
  \pgfsetlayers{background,main,foreground}
  \tikzstyle{vec}=[circle,inner sep=1pt,outer sep=-1pt,fill]
  \tikzstyle{border}=[thick]
  \tikzstyle{favborder}=[border,dotted]
  \tikzstyle{exclborder}=[border,dashed]
\usepackage{etex,etoolbox}
\usepackage{hyperref}

\hypersetup{colorlinks=true,linkcolor=blue,citecolor=magenta}

\usepackage{etex,etoolbox}
\usepackage{hyperref}

\newenvironment{proofof}[1]{\begin{trivlist}\item[\hskip\labelsep{\textsc{Proof
  of {#1}.\ }}]}{\hspace*{\fill} {$\square$}\end{trivlist}}

\newenvironment{tbs}{%
   \small\tt
   \begin{enumerate}[$\blacktriangleright$]}{\end{enumerate}}
\newcommand{\btbs}{\begin{tbs}}                                                                      
\newcommand{\etbs}{\end{tbs}}

\newcommand{\Reals}{\mathbb{R}}

\DeclareMathOperator{\posi}{posi}

  \DeclareMathOperator{\cl}{cl}
\newcommand{\pspace}{\Omega}

\newcommand{\reals}{\Reals}

\newcommand{\gambles}{\mathcal{L}}

\newcommand{\theory}{\mathfrak{T}}
\newcommand{\btheory}{\mathfrak{T}^\star}
\newcommand{\domain}{\mathcal{K}}
\newcommand{\assess}{\mathcal{G}}
\newcommand{\bdomain}{\mathcal{C}}

\newcommand{\nonnegative}{\gambles^{\geq} }
\newcommand{\bnonnegative}{\Sigma^{\geq}}

\newcommand{\negative}{\gambles^{<} }
\newcommand{\bnegative}{\Sigma^{<} }

\newtheorem{assumption}{Assumption}

\usepackage{isipta2019}
\usepackage[american]{babel}

\usepackage{booktabs} 
\usepackage{tikz} 

\title[Bernstein's socks and P-coherence]{Bernstein's socks and polynomial-time provable coherence}

\author{
   \Name{Alessio Benavoli}\Email{alessio.benavoli@ul.ie}\\
     \addr   Department of Computer Science and Information Systems (CSIS), University of Limerick, Ireland.\\
  \Name{Alessandro Facchini}\Email{alessandro.facchini@idsia.ch }\\
  \Name{Marco Zaffalon}\Email{zaffalon@idsia.ch}\\
  \addr  Dalle Molle Institute for Artificial Intelligence (IDSIA), Galleria 2, 6928 Manno Switzerland
}

\begin{document}

\maketitle

\begin{abstract}
We recently introduced a bounded rationality approach for the theory of desirable gambles. It is based on the unique requirement that being nonnegative for a gamble has to be defined so that it can be provable in polynomial-time. In this paper we continue to investigate  properties of this class of models. In particular we verify that the space of Bernstein polynomials in which nonnegativity is specified by the Krivine-Vasilescu  certificate is yet another instance of this theory. As a consequence, we show how it is possible to construct in it a thought experiment uncovering entanglement with classical (hence non quantum) coins.
\end{abstract}

\section{Introduction}
In a recent paper \citep{BenavoliCC}, we have shown that Quantum Theory (QT) is a theory of bounded rationality \citep{simon1957models} based on a different notion of nonnegativity. This is tantamount to change the class of gambles that should always be desired in such a way that the consistency problem becomes \textit{provable in polynomial-time} (we have called it \textit{P-coherence}).
Conversely, in the same settings, classical probability (standard ``almost desirability'' \citep{walley91}) is  NP-hard.
As a consequence, we have thence proved that the only physics' axiom in QT is computational tractability, yielding all its weirdness (different logic of events, negative probabilities, and entanglement). 

Interestingly, it turns out that entanglement is not peculiar to QT but an inherent characteristic of  bounded rationality for desirable gambles based on P-coherence, a model first introduced in \cite{Benavoli2017b} and implemented using sum-of-squares
polynomials in the real numbers.

A first goal of the present paper is to better understand the structural properties of P-coherence.  An elegant way of doing this is to look at the linear space of gambles ($\gambles$) as an algebra of formulas, and thus define a logic on it. In doing so, we verify that a sufficient condition for the reduction of  P-coherence to classical logical consistency (that is  the existence of a non derivable formula from the considered set of assessments) is for the set of tautologies to satisfy a certain ``pullup'' property.

As a second goal of our work, we provide yet another instance of a P-coherence model by using the so-called Krivine-Vasilescu nonnegativity certificate for polynomials. 
We show that, by focusing on polynomials defined on the \textit{simplex of probability},  this notion of nonnegativity defines the so-called \textit{Bernstein nonnegative polynomials}, a class whose cone, when polynomials are of any degree $d$,  has been introduced in \cite{de2009exchangeable}
to generalise de Finetti's representation result for exchangeable events (see also \citep{de2012exchangeability,de2012imprecise,DeCooman2015JAIR}).

Finally, we restrict our attention to  a finite degree $d$, compute the dual of a (Bernstein) P-coherent set of desirable gambles and illustrate how P-coherence can provide an example of \textit{entanglement with classical coins}: two coins that always land HH or TT, but for which it is not possible to find a ``common cause'' \citep{einstein1935can} (a classical correlation model) that explains
these results. The latter is done by deriving a Bell's type inequality in the Bernstein's world.

The title of  this paper is freely inspired by Bell's work ``Bertlmann's socks and the nature of reality'' \citep{bell1981bertlmann}.

\section{The logic of desirability}\label{sec:tdg}

Let $\pspace$ denote the possibility space of an experiment (e.g., $\{Head, Tail\}$ 
or $\reals^n$). A gamble $g$ on $\pspace$ is a bounded real-valued function of $\pspace$, interpreted as
an uncertain reward. 
Accepting a gamble $g$ by an agent, Alice, is regarded as a
commitment to receive, or pay (depending on the sign), $g(\omega)$ \emph{utiles}\footnote{Abstract units of utility, we can approximately identify it with money provided we deal with small amounts of it \cite[Sec.~3.2.5]{finetti1974}} whenever $\omega \in \Omega$ occurs. We denote  by $\gambles$ the set of all the gambles on $\pspace$, the subset of all nonnegative gambles, that is, of gambles for which Alice is never expected to lose utiles, is denoted as  
$\nonnegative\coloneqq \{g \in \gambles: \inf g\geq0 \}$ (analogously negative gambles are denoted as $\negative\coloneqq \{g \in \gambles: \sup g < 0 \}$).
In the following, with $\mathcal{G} \subset \gambles$ we denote a finite set of gambles that Alice finds desirable (we will comment on the case when $\mathcal{G}$ may not be finite): these are the gambles that she is willing to accept and thus commits herself to the corresponding transactions.

The crucial question is now to  provide a criterion for a set $\mathcal{G}$ of gambles representing assessments of desirability to be called \emph{rational}. Intuitively Alice is rational if she avoids sure losses (also-called Dutch books or arbitrages): that is, if, by considering the implications of what she finds desirable, she is not forced to find desirable a negative gamble. 
An elegant way to formalise this intuition is to see $\gambles$ as an algebra of formulas, and thus define a logic on it. Based on it, we can thus formulate rationality as \textit{logical consistency}. 

In the theory of (almost) desirable gambles, we may proceed as follows. First of all we introduce some basic notions from logic.


A sequent is a pair $(\assess, g)$, written $\assess \sequent g$, where $\assess$ is a set (possibly empty) of gambles, and $g$ is a gamble.  

One can read a sequent $\assess \sequent g$ as saying ``whenever the gambles in  $\assess$ are desirable (accepted) for (by) Alice, the gamble $g$ is also desirable (accepted) for (by) Alice''.

A Gentzen-style rule is a pair which consists of a set $\{\assess_i \sequent g_i \mid i \leq \alpha\}$ of sequents, called the premisses of the rule, and a sequent $\assess \sequent g$, called the consequence of the rule and therefore which follows from the set according to the rule. We let $\alpha\leq \omega$, and thus do not rule out the fact that a rule may be infinitary.
A rule $r$ is  written symbolically in the form

\begin{prooftree}
              \AxiomC{$\{\assess_i \sequent g_i \mid i \in \alpha\}$}
         \LeftLabel{\tiny (r) }
         \RightLabel{,}
              \UnaryInfC{$\assess \sequent g$}
     \end{prooftree}
\noindent which can be read as ``if  Alice accepts the gamble  $g_i$, given the fact that she accepts the gambles  in  $\assess_i$ (for $i \in \alpha$), then necessarily whenever the gambles in $\assess$ are accepted, the gamble $g$ is also accepted''.

An axiom is a rule in which the set of premisses is empty.

A system $\mathfrak{S}$ is a set of Gentzen-style rules. 
We say that a sequent $\assess \sequent g$ is provable in a system $\mathfrak{S}$ from a set of sequent $\{\assess_i \sequent g_i \mid i \leq \alpha\}$ if there is a well-founded tree whose leaves are labelled either with axioms or with members of $\{\assess_i \sequent g_i \mid i \leq \alpha\}$, whose root is labelled with $\assess \sequent g$ and the labelling of all nodes is consistent with the rules of $\mathfrak{S}$.   A sequent is provable $\mathfrak{S}$ if it is provable from the empty set.
 The set of all gambles $g$ such that $\assess \sequent g$, for some set $\assess$, 
 will sometimes be denoted by $\closure_\mathfrak{S} (\assess)$. Finally, we say that a set $\assess$ is \emph{consistent in  $\mathfrak{S}$} if there is a gamble $g$ such that $\assess \sequent g$ is not provable in $\mathfrak{S}$, that is 
\begin{equation}
\label{eq:consistent}
\exists g\in \gambles ~\text{ such that }~ g\notin \closure_\mathfrak{S} (\assess).
\end{equation}
A set $\assess$ is closed in $\mathfrak{S}$ whenever $\closure_\mathfrak{S} (\assess)= \assess$, and is called a theory of $\mathfrak{S}$. In general, theories are denoted by $\domain$. 
When $\closure_\mathfrak{S}$ is a consequence operator, that is it is reflexive, monotone and transitive, the consistent theories of $\mathfrak{S}$ completely characterised it, in the sense that $\closure_\mathfrak{S} (\assess) $ coincides with the intersection of all consistent theories of $\mathfrak{S}$ extending $\assess$.

The system $\theory$ for the theory of (almost) desirable gambles is thus defined as follows:

{\small
\vspace{0.3cm}
\setlength{\tabcolsep}{0pt}
\begin{tabular}{ll}
  \multicolumn{2}{l}{
  \textbf{Structural axiom:}}\vspace{3mm}\\
  \textit{Reflexivity} 
  \\
              & {\AxiomC{}
         \LeftLabel{\tiny (R) }
         \RightLabel{, for  $g \in \assess$}
              \UnaryInfC{$\assess  \sequent g$}
              \DisplayProof}\vspace{2mm}\\
  \multicolumn{1}{l}{\textbf{Logic axiom:}}\vspace{3mm}\\
\textit{Accept. nonneg.} 
\\
& {              \AxiomC{}
         \LeftLabel{\tiny (ANN) }
         \RightLabel{, for  $g \in \nonnegative$}
              \UnaryInfC{$\assess \sequent g$}
              \DisplayProof}\vspace{2mm}\\
                \multicolumn{1}{l}{\textbf{Logic rules:}}\vspace{3mm}\\              
\textit{Conic hull} 
&{     \AxiomC{$\assess \sequent g$}
     \AxiomC{$\assess' \sequent f$}
  \LeftLabel{\tiny (L) }
   \RightLabel{
  , for  $\mu, \lambda \geq 0$
   }
     \BinaryInfC{$\assess, \assess' \sequent \mu g+ \lambda f$}
                  \DisplayProof}\vspace{2mm}\\
                  \textit{Closure} 
                  \\
&{       \AxiomC{$\{\assess \sequent g + \epsilon^\ell \mid \ell > 0 \} $}
  \LeftLabel{\tiny (C) }
  \RightLabel{, for $\epsilon \in (0,1)$}
    \UnaryInfC{$\assess \sequent g$}
                  \DisplayProof}\vspace{3mm}\\
\end{tabular}
 }

It is easy to verify that $\closure_\theory$ is indeed a consequence operator on $\gambles$, and thus in particular that the usual structural rules of Gentzen systems such as \textit{weakening} and \textit{cut} are derivable in $\theory$.
It is immediate to verify that logical consistency in $\theory$ is indeed tantamount to rationality (coherence, no arbitrage):

\begin{theorem}\label{prop:cons}
 Let  $\assess$ be a set of assessments. It holds that  $\closure_\theory (\assess) = \cl(\posi(\assess))$, where $\posi$ denotes the conic hull operators and $\cl$ the topological closure operator.
Moreover the following conditions are equivalent
 \begin{enumerate}
 \item $\assess$ is logically consistent in $\theory$ 
 \item $\assess$ avoids negativity (sure loss), that is $\negative \cap \closure_\theory(\assess) = \emptyset$ 
 \item $-1 \notin  \closure_\theory (\assess)$.
 \end{enumerate}
\end{theorem}
Therefore,  $-1$ can be regarded as playing the role of the (classical) \textit{Falsum}.
  
Clearly when $\mathcal{G}$ is finite, $\closure_\theory(\assess)$ simply coincides with the conic hull closure of $\assess$. From Theorem \ref{prop:cons} we thence obtain that the \textit{principle of explosion} is derivable in $\theory$:
\vspace{0.5cm}

{\small
\begin{tabular}{ll}               
 &{     \AxiomC{$\assess \sequent -1$}
  \LeftLabel{
  \tiny (Explosion) 
  }
   \RightLabel{
, for  $g \in \gambles$.
  }
     \UnaryInfC{$\assess \sequent g$}
       \DisplayProof}\vspace{8mm}\\
\end{tabular}
}

By observing that the mathematical dual of $\domain$ is a closed convex set of probabilities, we can then provide a \textit{semantic} (probabilistic interpretation) to  $\theory$:
\begin{equation}
\label{eq:dual}
\begin{aligned}
 \mathcal{P}(\assess)=\left\{\mu  \in \mathsf{S} \Big| \int_{\pspace} g(\pspace) d\mu(\pspace)\geq0, ~\forall  g\in \mathcal{G}\right\},\\
\end{aligned}
\end{equation}
where $\mathcal{S}=\{ \mu \in \mathcal{M} \mid \inf \mu \geq0,~\int_{ \pspace}  d\mu(\pspace)=1\}$  is the set of all probabilities in $\pspace$, also-called (belief) \emph{states}, and $ \mathcal{M}$ the set of all charges (a charge is a finitely additive signed-measure \cite[Ch.11]{aliprantisborder}) on $\pspace$.  
The duality  actually provides us immediately with a sound and completeness results. Indeed, say that a state  $\mu \in \mathcal{S}$ is a model of a set of gambles $\assess$ whenever $\mu \in  \mathcal{P}(\assess)$. Then the following is an immediate consequence of Theorem \ref{prop:cons} and \cite[Theorem 4]{pmlr-v62-benavoli17b}. 
\begin{theorem}\label{theo:com}
For every set of gambles $\assess \cup \{g\} \subseteq \gambles$, it holds that
\begin{equation}\label{eq:complet}
g \in \closure_{\theory}(\assess)\iff \mathcal{P}(g) \subseteq \mathcal{P}(\assess). \end{equation}
In particular $\assess$ is inconsistent iff $\mathcal{P}(\assess) = \emptyset$ (it has no probabilistic model).
\end{theorem}

Hence, whenever an agent is coherent (that is the rationality of her behaviour is represented by a logical consistent theory in $\theory$), Equation~\eqref{eq:dual} states that desirability corresponds to nonnegative expectation (for all probabilities in $\mathcal{P}$). When she is incoherent,  $\mathcal{P}$ turns out to be empty.

\begin{remark}\label{remark}
In \eqref{eq:consistent} we have defined the notion of logical consistency for a theory as being non trivial (different from the whole language) and thus, from a semantic perspective, as having a model (that is $\mathcal{P}$ is not empty).
Theorem \ref{prop:cons} attests that $-1$ is tantamount to the (classical) Falsum, that is an all implying formula that has no model. 

The importance of being able to reduce incoherence to logical inconsistency can be appreciated by the following argument. Assume this is not the case,  the explosion principle for $-1$ does not hold. In particular, this means that we may be able to find two different incoherent sets  of  gambles (formulas). But then such sets cannot be separated in the dual space, meaning that duality would fail in providing us with a sound and  complete probabilistic semantics for the system  under consideration (that is satisfying the correspondence in Equation \eqref{eq:complet} of  Theorem \ref{theo:com}).

The accent in this work to the capability of reducing coherence (also when formulated as P-coherence) to logical consistency is therefore justified by the fact that we do not want a situation of  incoherence (irrationality) to represent anything else than a situation of incoherence (irrationality), and thus for which, seen as a theory of a logic, there is no ``natural'' model (in the dual space). 
\end{remark}

\section{The complexity of inference}\label{sec:comp}

In light of Theorem \ref{prop:cons}, when the theory is finitely generated, the problem of checking whether or not $\domain$ is consistent can be formulated as the following decision problem:
\begin{equation}
\label{eq:dec}
\begin{aligned}
\exists \lambda_i\geq0:-1-\sum\limits_{i=1}^{|\mathcal{G}|} \lambda_i g_i \in \gambles^{\geq}.
\end{aligned}
\end{equation}
If the answer is ``yes'', then the gamble $-1$ belongs to $\domain$, proving $\domain$'s inconsistency. Actually any inference task can ultimately be reduced to a problem of the form~\eqref{eq:dec}: the lower prevision (expectation) of a gamble $q$ is $\underline{E}(q)=\sup_{\lambda_0\in\reals,\lambda_i\in \reals^{\geq}} \lambda_0:~ q-\lambda_0-\sum_{i=1}^{|\mathcal{G}|} \lambda_i g_i \in \gambles^{\geq}$.
Hence, the above decision problem  unveils  a crucial fact: the hardness of inference in classical probability corresponds to the hardness of evaluating the nonnegativity of a function in the considered space (the ``nonnegativity decision problem'').

When $\pspace$ is infinite, and for generic functions, the nonnegativity decision problem is \textit{undecidable}. To avoid such an issue, we may impose  restrictions on the class of allowed gambles. For instance, instead of $\gambles$, we may consider $\gambles_R$: the class of multivariate polynomials of degree at most $d$ (we denote by $\gambles_R^{\geq}\subset\gambles_R$ the subset of nonnegative polynomials and by $\gambles_R^{<}\subset\gambles_R$ the negative ones).  In  doing so, by Tarski-Seidenberg quantifier elimination theory \citep{tarski1951decision,seidenberg1954new}, the decision problem becomes decidable, but still intractable, being in general NP-hard. If we  accept that P$\neq$NP and we require that inference should be tractable  (in P), we are stuck. 
What to do? A solution is to change the meaning of ``being nonnegative'' for a  function by considering a subset $\bnonnegative \subsetneq \nonnegative_R$ for which the membership problem in \eqref{eq:dec} is in P. 

In other words, a computationally efficient version of TDG, which we denote by $\btheory$, 
should be based on a redefinition of the logical axiom scheme , i.e., by stating that

\vspace{0.2cm}
{\small
\begin{tabular}{ll}               
\textit{Accept. P-nonneg.}\\ 
& {              \AxiomC{}
         \LeftLabel{\tiny (P) }
           \RightLabel{
          , for  $g \in \bnonnegative$.
           }
              \UnaryInfC{$\assess \sequent g$}
              \DisplayProof}\vspace{2mm}\\
\end{tabular}
}

We thus denote by $\btheory$ the  Gentzen system (on $\gambles_R$) obtaining from $\theory$  by substituting (P) to (ANN), and call it a \textit{P-system}.

\section{Coherence and semantics}\label{sec:con}

\subsection{P-coherence vs consistency}\label{sub:con}

However, how can be sure that we have done things properly, that $\btheory$ is really just a  computationally efficient version of $\theory$?
In order to do so we would like to verify that, for finite sets $\assess$,  logical consistency can be checked in polynomial-time by solving:
\begin{equation}
\label{eq:bdec}
\begin{aligned}
\exists \lambda_i\geq0 ~~\text{ such that }~~ -1-\sum\limits_{i=1}^{|\mathcal{G}|} \lambda_i g_i \in \bnonnegative.
\end{aligned}
\end{equation}
Note that, the lower prevision of a gamble $q$ in this case is
\begin{equation}
\label{eq:lowerprev}
\begin{aligned}
&\underline{E}_B(q)=\sup_{\lambda_0\in\reals,\lambda_i\in \reals^{\geq}} \lambda_0\\
&s.t.\\
&q-\lambda_0-\sum\limits_{i=1}^{|\mathcal{G}|} \lambda_i g_i \in \bnonnegative.
\end{aligned}
\end{equation}

To verify that logical consistency can be checked in polynomial-time, we need to find an analogous of Theorem \ref{prop:cons}, but clearly also to assume that $\gambles_R$ contains all constant gambles.

In what follows, we provide some sufficient conditions for this to hold. 
First of all, it is reasonable to ask to the new variant of ``being nonnegative''  (that is to the set $\bnonnegative$) to be a closed convex cone. 
Avoiding nonnegativity (that is coherence, rationality, no arbitrage) can now be redefined as follows
\begin{definition}[P-coherence]
\label{def:bavs}
A set $\bdomain \subseteq \gambles_R$ is \emph{P-coherent}  if  $\bnegative \cap  \bdomain=\emptyset$, where $\bnegative$ is the interior of $\{g  \mid -g \in \bnonnegative\}$.
\end{definition}
From now on, we also always make another minimal reasonable assumption of $\bnegative$ being non empty.
 The next result states essentially that $-1$ represents  P-incoherence.

  \begin{proposition}\label{prop:cohe}
Let $\assess \subseteq \gambles_R$ a set of assessments. It holds that $\cl\posi(\assess)=\closure_{\btheory}(\assess)$.
Moreover the following are equivalent: 
\begin{enumerate}
 \item $-1 \notin \posi ( \assess \cup \bnonnegative)$
    \item $\posi ( \assess \cup \bnonnegative)$ is P-coherent 
     \item $\closure_{\btheory}(\assess)$ is P-coherent
\end{enumerate}
  \end{proposition}

  Analogously with $\theory$, one can ask whether the class of P-coherent theories characterise the system $\btheory$.  It turns out that if we want to be sure this to  be the case (see Remark \ref{remark} for the reason why we want this) we need to add some structure to $\bnonnegative$. For instance:
  
  \begin{proposition}\label{prop:eqq}
Let $\assess \subseteq \gambles_R$ a set of assessments and assume that 
\begin{description}
\item[(pullup)] 
for every $f \in \gambles_R$, there is $\epsilon > 0$ such that $f+\epsilon \in \bnonnegative$.
\end{description}
 The following are then equivalent
\begin{enumerate}
     \item $\closure_{\btheory}(\assess)$ is P-coherent
    \item $\closure_{\btheory}(\assess) $ is logically consistent.
\end{enumerate}
  \end{proposition}

Hence, whenever (*) holds and Proposition \ref{prop:eqq} can be applied, logical consistency in $\btheory$ for finitely generated theories can be checked efficiently.
This is the case of QT and of the family of P-systems defined on Bernstein polynomials introduced in the following sections of this work.


\subsection{Probabilistic interpretation of P-systems}

Interestingly, we can associate a ``probabilistic'' interpretation as before to the system $\btheory$ by computing the dual of a theory. Since $\gambles_R$ is a topological vector space, we can consider its dual space $\gambles_R^*$ of all bounded linear functionals $L: \gambles_R \rightarrow \reals$. Hence, with the additional condition that linear functionals preserve the unitary gamble, the dual cone of a theory $\bdomain\subset \gambles_R$  is given by
\begin{equation}
\label{eq:dualL}
\bdomain^\circ=\left\{L \in \gambles_R^* \mid L(g)\geq0, ~~L(1)=1,~\forall g \in \bdomain\right\}.
\end{equation} 
Based on Equation \ref{eq:dualL} and its properties, under the pullup assumption one then gets the analogous of Theorem \ref{theo:com} but for $\btheory$.
The question now is whether we can ``massage'' this result and obtain a sound and complete classical probabilistic semantics.
In this aim, first notice that to $\bdomain^\circ$ we can  associate its extension  $ \bdomain^\bullet$ in $\mathcal{M}$, that is, the set of all  charges
 on $\pspace$ extending an element in $\bdomain^\circ$. 
However, as shown in  \citep{BenavoliCC},  in doing this one cannot in general provide an adequate classical probabilistic interpretations to $\btheory$, except if one allow for instance the use of quasi-probabilities (probability distributions that admit negative values). This is essentially due to the fact that whenever
$\bnegative \subsetneq \negative_R$, there are negative gambles that cannot be proved to be negative in polynomial-time. This observation (made mathematically precise in \cite[Theorem 1]{BenavoliCC}) provides for instances an explanation  of all paradoxes of quantum mechanics, a special instance of a P-system. 

\section{Krivine-Vasilescu's nonnegativity}
Let $\gambles_R$ be the space of all polynomials of $n$ variables of degree $R$ in $\Omega$.\footnote{The $R$ in $\gambles_R$ should stay for ``Restricted'', here we also use it to denote  degree of the polynomial $R$.}
Let us  assume that $\Omega \subset \reals^n$ is a compact semi-algebraic set, i.e., a compact set described by polynomial inequalities
\begin{equation}
\label{eq:Omegaconstr}
 \Omega=\left\{x \in \reals^n:  ~ c_j(x)\geq0, ~~j =1,\dots,m \right\}.
\end{equation}

Let $\bar{c}_j$ be equal to $\sup_{x\in  \Omega} c_j(x)$, define
\begin{equation}
\label{eq:normC}
\hat{c}_j(x)=\left\{\begin{array}{lll}
               c_j(x)/\bar{c}_j &\text{ if } &\bar{c}_j>0,\\
                c_j(x)          &\text{ if } &\bar{c}_j=0.\\
              \end{array}\right.
\end{equation}
Therefore, it results that $\hat{c}_j(x)\geq0$ and $1-\hat{c}_j(x)\geq0$ for each $x\in  \Omega$.
Consider the closed convex cone

\begin{equation}
 \label{eq:b0KV}
\resizebox{.95\hsize}{!}{$\bnonnegative_d =\Bigg\{  \sum\limits_{({\alpha},{\beta})\in  \mathbb{N}^{2m}_{d}} \hspace{-5mm}u_{{\alpha}{\beta}} 
  \hat{c}_{1}^{\alpha_1}\cdots \hat{c}_{m}^{\alpha_m} (1-\hat{c}_{1})^{\beta_1}\cdots (1-\hat{c}_{m})^{\beta_m} :  u_{{\alpha}{\beta}}\in \mathbb{R}^{\geq} \Bigg\}$},
 \end{equation}
where   $ \mathbb{N}^{2|c|}_{d}=\{({\alpha},{\beta})\in\mathbb{N}^{2|c|} : |{\alpha}+{\beta}|\leq d\}$ and
  $|{\alpha}|=\sum_{i=1}^m \alpha_i$. We denote with $R$ the maximum degree of the polynomials in $\bnonnegative_d$, so that
  $\bnonnegative_d\subset \gambles_R$.

\begin{assumption}
We assume that $  \bnonnegative_d $ satisfies the ``pullup'' property for every $d\in \mathbb{N}$.
\end{assumption}

We can then define the Krivine-Vasilescu nonnegativity certificate \citep{krivine1964anneaux,vasilescu2003spectral,lasserre2009moments}.
\begin{definition}[Krivine-Vasilescu]
A polynomial of $g \in \gambles_R$ is ``nonnegative'' in $\pspace$ when it belongs to $\bnonnegative_d$
defined in \eqref{eq:b0KV}.
\end{definition}
That is, the function is ``nonnegative'' whenever the coefficients $u_{{\alpha}{\beta}}$ are nonnegative.

\begin{proposition}
Given definition of $\bnonnegative_d$ in \eqref{eq:b0KV}, the consistency problem \eqref{eq:bdec} can be solved in polynomial-time
\end{proposition}
The proof is immediate and consists in showing that the membership problem $g \in \bnonnegative_d$ can be formulated as a \textit{linear programming} problem.
We will give an example of that in the next section.  
Let $\mathcal{G}$ be a finite set of assessments, and $\bdomain$ its deductive closure $\text{posi}(\mathcal{G}\cup \bnonnegative_d)$, with the given definition  $\bnonnegative_d$ in \eqref{eq:b0KV}, satisfying $-1 \notin \bdomain$. By Proposition \ref{prop:cohe}  $\bdomain$ is P-coherent, and therefore also logical consistent.


Moreover, it is not difficult to prove that the dual of $\bdomain$  is
\begin{align}
\label{eq:credaldef}
\mathcal{Q}=\Big\{L \in  \mathsf{S} \mid ~L(g)\geq0, ~~\forall g \in \mathcal{G}\Big\},
\end{align} 
with the set of (belief) \textit{states} defined as:
\begin{equation}
\label{eq:bernsteinstates}
\begin{aligned}
\mathsf{S}=\Big\{L \in \gambles_R^* \mid L(1)=1,\\
L(\hat{c}_{1}^{\alpha_1}\cdots \hat{c}_{m}^{\alpha_m} (1-\hat{c}_{1})^{\beta_1}\cdots (1-\hat{c}_{m})^{\beta_m})\geq0,\\
\forall~({\alpha},{\beta})\in  \mathbb{N}^{2|c|}_{d}\Big\}.
\end{aligned}
\end{equation}
Note that, the linear operator acts on the monomials (this follows by linearity) and, therefore, the dual space is isomorphic
to $\reals^{s_n(d)}$, with $s_n(d)$ being the number of all monomials for a generic polynomial of $n$ variables and degree $d$.
That means we can define the real numbers
\begin{equation}
\label{eq:linearoperator}
z_{\gamma_1\gamma_2\dots \gamma_n}=L(x_1^{\gamma_1}x_2^{\gamma_2}\cdots x_n^{\gamma_n})  \in \reals,
\end{equation} 
with $\gamma_i \in \mathbb{N}$, and  we can rewrite 
$L(f)$, for any polynomial $f \in \gambles_R$, as a function of the vector of variables $z\in \reals^{s_n(d)}$, whose components are the real variables $z_{\gamma_1\gamma_2\dots \gamma_n}$ defined above. 


\subsection{Simplex of probability}
A particular case of the Krivine-Vasilescu's nonnegativity criterion is obtained  when $ \Omega$ is the probability simplex: 
\begin{equation}
\label{eq:OmegaconstrStheta}
 \Omega=\left\{\theta \in \reals^n:  ~ \theta_j\geq0, ~~1-\sum_{j=1}^n \theta_j\geq 0 \right\},
\end{equation}
where the changed notation, $\theta$ instead of $x$, reflects the fact that the variables are probabilities.\footnote{
Note that, in this case $m=n+1$ in \eqref{eq:Omegaconstr}.}
In this case, we can simplify the definition of  nonnegativity  in \eqref{eq:b0KV}  as \cite[Sec.5.4.1]{lasserre2009moments}:
\begin{equation}
 \label{eq:b0simpl0}
 \begin{aligned}
 \resizebox{.95\hsize}{!}{$\bnonnegative_d =\Bigg\{  \sum\limits_{\alpha\in  \mathbb{N}^{n+1}_{d},|\alpha|\leq d} \hspace{-3mm}u_{{\alpha}} 
  \theta_{1}^{\alpha_1}\cdots  \theta_{n}^{\alpha_{n}}(1-\theta_1-\dots-\theta_n)^{\alpha_{n+1}} :  u_{{\alpha}}\in \mathbb{R}^{\geq}  \Bigg\}$}.
 \end{aligned}
 \end{equation}
 Note that, the maximum degree of the polynomials in $\bnonnegative_d$ is $d$ and, therefore, $R=d$.
 It is easy to prove the following.
 \begin{proposition}
 \label{eq:proppullBern}
   $  \bnonnegative_d $ satisfies the ``pullup'' property for every $d\in \mathbb{N}$.
 \end{proposition}
 We have also that  \eqref{eq:bernsteinstates} in this case becomes:  
\begin{equation}
\label{eq:bernsteinstates1}
\begin{aligned}
\mathsf{S}=\Big\{L \in \gambles_R^* \mid L(1)=1,\\
 L( \theta_{1}^{\alpha_1}\cdots  \theta_{n}^{\alpha_{n}}(1-\theta_1-\dots-\theta_n)^{\alpha_{n+1}})\geq0,\\
\forall~\alpha\in\mathbb{N}^{n+1}_{d},|\alpha|\leq d\Big\}.
\end{aligned}
\end{equation}
 
We recall that the Bernstein (multivariate) polynomials of degree $d$ are:
$$
B_{\gamma,d}(\theta)={d \choose \gamma}\theta_1^{\gamma_1}\dots \theta_{n}^{\gamma_n}(1-\theta_1-\dots-\theta_n)^{d-|\gamma|},
$$
where $\gamma=(\gamma_1,\dots,\gamma_n)$ with $\gamma_i\in\mathbb{N}$. Moreover, all Bernstein  polynomials of fixed degree $d$ forms a
basis for the linear space of all polynomials whose degree is at most $d$, and  they form a partition of unity:
$\sum_{\gamma:|\gamma|=d}B_{\gamma,d}(\theta)=1$.
By exploiting these properties, we can  prove the following.
\begin{proposition}
\label{prop:coneBern}
 A polynomial $g:\Omega \rightarrow \reals$ of degree $d$ belongs to the cone in \eqref{eq:b0simpl0} iff it belongs to
\begin{equation}
 \label{eq:b0simpl}
 \begin{aligned}
   \resizebox{.95\hsize}{!}{$\widetilde{\Sigma}^{\geq}_{d} =\Bigg\{  \sum\limits_{\alpha\in  \mathbb{N}^{n+1}_{d},|\alpha|=d} \hspace{-3mm}u_{{\alpha}} 
  \theta_{1}^{\alpha_1}\cdots  \theta_{n}^{\alpha_{n}}(1-\theta_1-\dots-\theta_n)^{\alpha_{n+1}} :  u_{{\alpha}}\in \mathbb{R}^{\geq}  \Bigg\}$}.
 \end{aligned}
 \end{equation}
\end{proposition}
 For this reason, we call $\widetilde{\Sigma}^{\geq}_{d}$ (equiv. ${\Sigma}^{\geq}_{d}$)  the cone of \textit{Bernstein nonnegative polynomials}
of degree $d$.\footnote{We could again simplify  \eqref{eq:bernsteinstates1}, but we leave the redundant formulation 
\eqref{eq:bernsteinstates1} because it requires to specify $L$ on the monomials too.}


In general, it holds that
$$
\widetilde{\Sigma}^{\geq}_{d} \subset \gambles_R^{\geq}
$$
that is, there exist nonnegative polynomials that are not included in $\widetilde{\Sigma}^{\geq}_{d}$.
\begin{example}
We use this counter-example from \citep{DeCooman2015JAIR}
 $$
 q(\theta)=\theta_1^2-\theta_1\theta_2+\theta_2^2,
 $$
 with $n=d=2$,  which is nonnegative in $\Omega$. 
 Now consider the  cone
\begin{equation}
\label{eq:coneposex}
 \resizebox{0.95\hsize}{!}{$
\begin{aligned}
&\widetilde{\Sigma}^{\geq}_{2}=\{u_{ijk}\in \reals^{\geq}: u_{002}(1-\theta_1 - \theta_2)^2 +  u_{011}\theta_2(1-\theta_1 - \theta_2)\\
&+ u_{020}\theta_2^2 +  u_{101}\theta_1(1-\theta_1 - \theta_2) + u_{110}\theta_1\theta_2 + u_{200}\theta_1^2\} 
\end{aligned}$}
\end{equation}
and an empty $\mathcal{G}$. The lower prevision of $q$ can be computed
 as follows:
\begin{equation}
\begin{aligned}
 \underline{E}_B(q)=\sup_{\lambda_0\in \reals,u_{ijk}\in \reals^{\geq}} \lambda_0\\
 -u_{002} + u_{101} - u_{200} + 1 =0\\
 -2u_{002} + u_{011} + u_{101} - u_{110} - 1=0\\
 2u_{002} - u_{101}=0\\
 -u_{002} + u_{011} - u_{020} + 1=0\\
 2u_{002} - u_{011}=0\\
 - \lambda_0 - u_{002}=0
\end{aligned}
\end{equation}
where the equality constraints have been obtained by equating the coefficients of the monomials
in 
$$
\begin{aligned}
q(\theta)-\lambda_0 = u_{002}(1-\theta_1 - \theta_2)^2 +  u_{011}\theta_2(1-\theta_1 - \theta_2)\\
+u_{020}\theta_2^2 +  u_{101}\theta_1(1-\theta_1 - \theta_2) + u_{110}\theta_1\theta_2 + u_{200}\theta_1^2.
\end{aligned}
$$
The solution of the above LP problem is
 $$\begin{aligned}
    &[\lambda_0,u_{002},u_{011},u_{020},u_{101},u_{110},u_{200}]\\
&= [-0.5, \, 0.5, \, 1 ,\,  1.5,\,  1 ,\,  0 , \, 1.5],
    \end{aligned}
 $$
 which means that $\underline{E}_B(q)=-0.5<0$.
 It is easy to verify that  $\underline{E}_B(q)<0$ for any $d\geq2$ and, therefore,  $q$ does not belong to $ \bnonnegative_d$ for any $d$.
 However, Figure \ref{fig:1} shows that $\underline{E}_B(q)$ quickly tends  to zero at the increase of the degree $d$ of $ \bnonnegative_d$.
 Therefore, we can build a  hierarchy of LPs \cite[Sec.5.4]{lasserre2009moments} of increasing size such that
 $$
 \underline{E}_B(q)  \xrightarrow{d\rightarrow\infty}  \underline{E}(q).
 $$
 This is true in general provided that $\mathcal{G}\cap \gambles_R^{<}=\emptyset$.
 \begin{figure}[htp!]
 \centering
  \includegraphics[width=7cm]{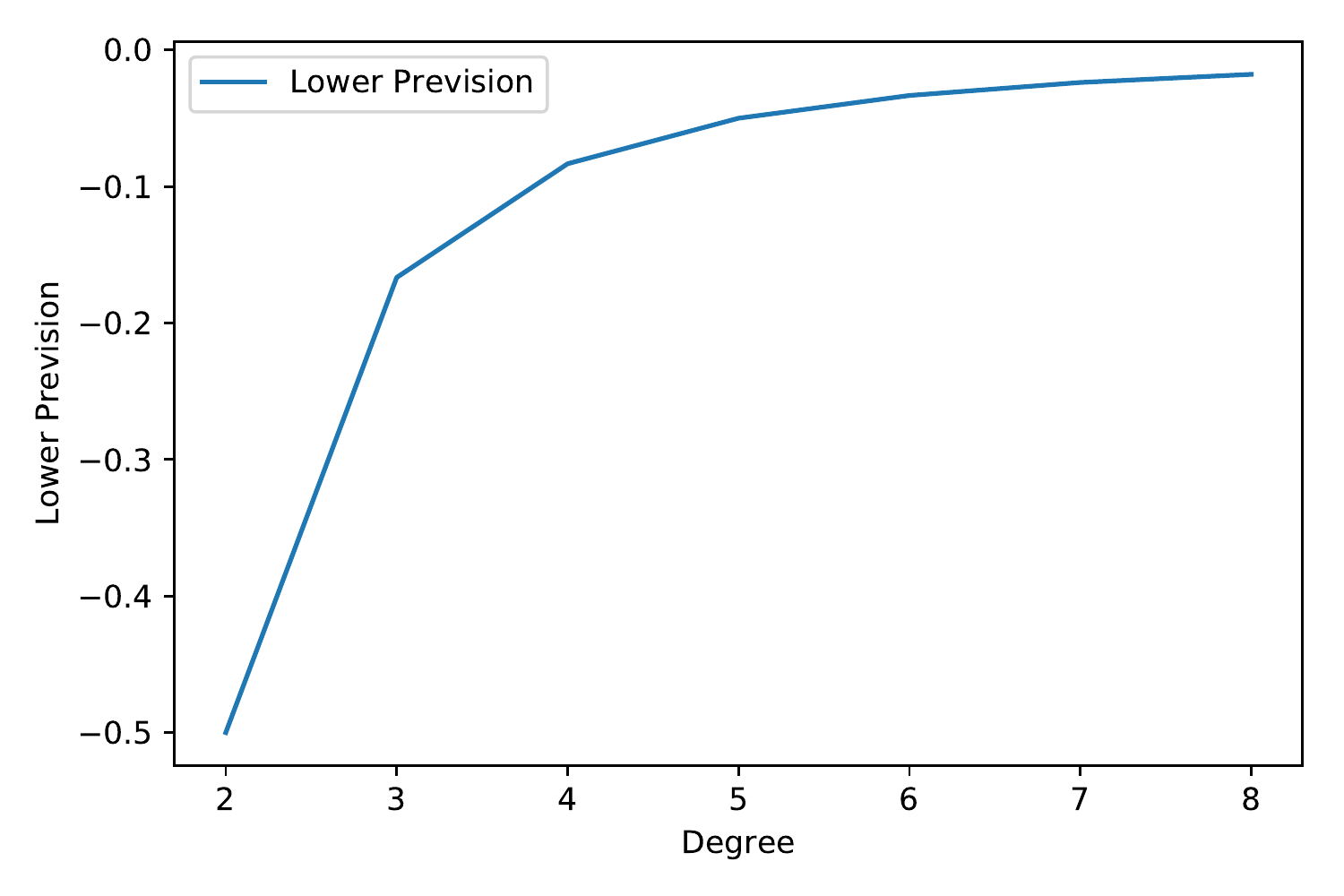}
  \caption{Convergence of $\underline{E}_B(q)$  to $\underline{E}(q)$ at the increase of the degree $d$.}
  \label{fig:1}
 \end{figure}
\end{example}

\subsection{Updating via partition of unity}
$\gambles_R$ is a space of polynomials and, therefore, it does not include indicator functions.
That means we cannot define conditioning.
However, we can still update our beliefs in a weaker way using a partition of unity.
We have seen that the Bernstein polynomials form a partition of unity, e.g.,
\begin{align}
\label{eq:Part1_d1}
n=3,d=1:~~~ & \{\theta_1,\theta_2,\theta_3,1-\theta_1-\theta_2-\theta_3\},
\end{align}
they are nonnegative functions and sum up to one.  We can then compute an updated lower prevision for a  gamble $q$ as: 

\begin{equation}
\label{eq:regularPost}
 \begin{array}{rlrl}
 \underline{E}_B(q|\pi)= &\sup\limits_{\lambda_j\geq0, \lambda_0} \lambda_0\\
  &s.t.\\
  &  (q-\lambda_0)\pi- \sum\limits_{j=1}^{|G|}\lambda_jg_j \in \bnonnegative_d,\\  
  \end{array}
\end{equation}
where $d$ must be large enough to guarantee that the membership problem is well-posed and $\pi$ denotes any \textit{subset sum } of the partition of unity
(in the example in \eqref{eq:Part1_d1}  $\pi\in\{\theta_1,\theta_2,\theta_3,1-\theta_1-\theta_2-\theta_3,\theta_1+\theta_2,
\theta_1+\theta_3,\theta_2+\theta_3,\dots\}$).
To alternatively justify this rule, we point out that $\pi$ can be interpreted as a  multinomial likelihood\footnote{The multinomial distribution is used to model the outcome of $\ell$ experiments, where the outcome of each trial has a categorical distribution, e.g., rolling a $k$-sided die $\ell$ times.} and the result of \eqref{eq:regularPost} as a bounded rationality version of a \textit{regular} posterior \cite[Appendix J5]{walley91}.

\subsection{A Bell inequality in the Bernstein world}
In this section, we derive a Bell's type inequality in the Bernstein world:
a probabilistic inequality that holds in $\theory$ but that is violated in $\theory^*$ (Bernstein world).
We will derive it by building a negative polynomial that has positive
prevision in $\theory^*$.
In the next section, we will show that the \textit{state} assigning a positive prevision
to such polynomial is \textit{entangled}!
%
For this purpose, we  consider  two coins, that we denote as $l$ (left) and, respectively, $r$ (right), and define
$$
\begin{bmatrix}
 \theta_1\\
 \theta_2\\
 \theta_3\\
 1-\theta_1-\theta_2-\theta_3
\end{bmatrix}
=\text{Prob}\begin{bmatrix}
 H_lH_r\\
 T_lH_r\\
 H_lT_r\\
 T_l,T_r\end{bmatrix},
$$
where $H_i,T_i$ denote the outcome Heads and, respectively, Tails for the left or right coin. In this case, the 
possibility space is
\begin{equation}
\label{eq:Omegaconstrex}
 \Omega=\left\{\theta \in \reals^3:  ~ \theta_1,\theta_2,\theta_3\geq0, ~~1-\theta_1-\theta_2-\theta_3\geq 0 \right\}.
\end{equation}
Note that, the following \textit{marginal} relationships hold:
$$
\begin{aligned}
\theta_{H_l}=\text{Prob}(H_l)=\theta_1+\theta_3,~~\theta_{H_r}=\text{Prob}(H_r)=\theta_1+\theta_2.
\end{aligned}
$$
Let $d$ be equal to $2$ and consider the \textit{state}:
\begin{equation}
\label{eq:state}
\begin{aligned}
 L(\theta_1)=z_{100}=1/3 & & L(\theta_1^2)&=z_{200}=1/3\\
 L(\theta_2)=z_{010}=1/6 & & L(\theta_2^2)&=z_{020}=0\\
 L(\theta_3)=z_{001}=1/6 & & L(\theta_3^2)&=z_{002}=0\\
  L(\theta_1\theta_2)=z_{110}=0 & & L(\theta_1\theta_3)&=z_{101}=0\\
   L(\theta_2\theta_3)=z_{011}=1/6 & & L(1)&=z_{000}=1,\\
\end{aligned} 
\end{equation}
which  belongs to \eqref{eq:bernsteinstates1}.

Now consider the  polynomial:
$$
q(\theta)=-(\theta_1 +\theta_2)^2 - (\theta_1 + \theta_3)(-2 \theta_1 - 2 \theta_2 + 1) -\epsilon,
$$
with $\epsilon>0$ and observe that $q(\theta)\leq -\epsilon$. 
The polynomial is negative! However, its prevision\footnote{The lower and upper previsions coincide for $q$.} $E_B(q)$ w.r.t.\ the state \eqref{eq:state} is equal to
\begin{align}
\nonumber
L(q)=-z_{001} + 2z_{011} - z_{020} - z_{100} + 2 z_{101} + z_{200}-\epsilon\\
=\frac{1}{6}-\epsilon\geq0.
\end{align} 
Therefore, we have violated an inequality that holds in classical probability  ($E(q)\leq -\epsilon$ in $\theory$), although the set of desirable gambles
$$
\bdomain=\{g \in \gambles_R \mid L(g)\geq0\},
$$
with $L$ defined in \eqref{eq:state}, is logical consistent in $\btheory$ (P-coherent).
This  is the essence of Bell's type inequalities: the quantum weirdness which  is also present in  Bernstein's world.
\begin{figure*}[htp!]
\centering
 \includegraphics[width=15cm]{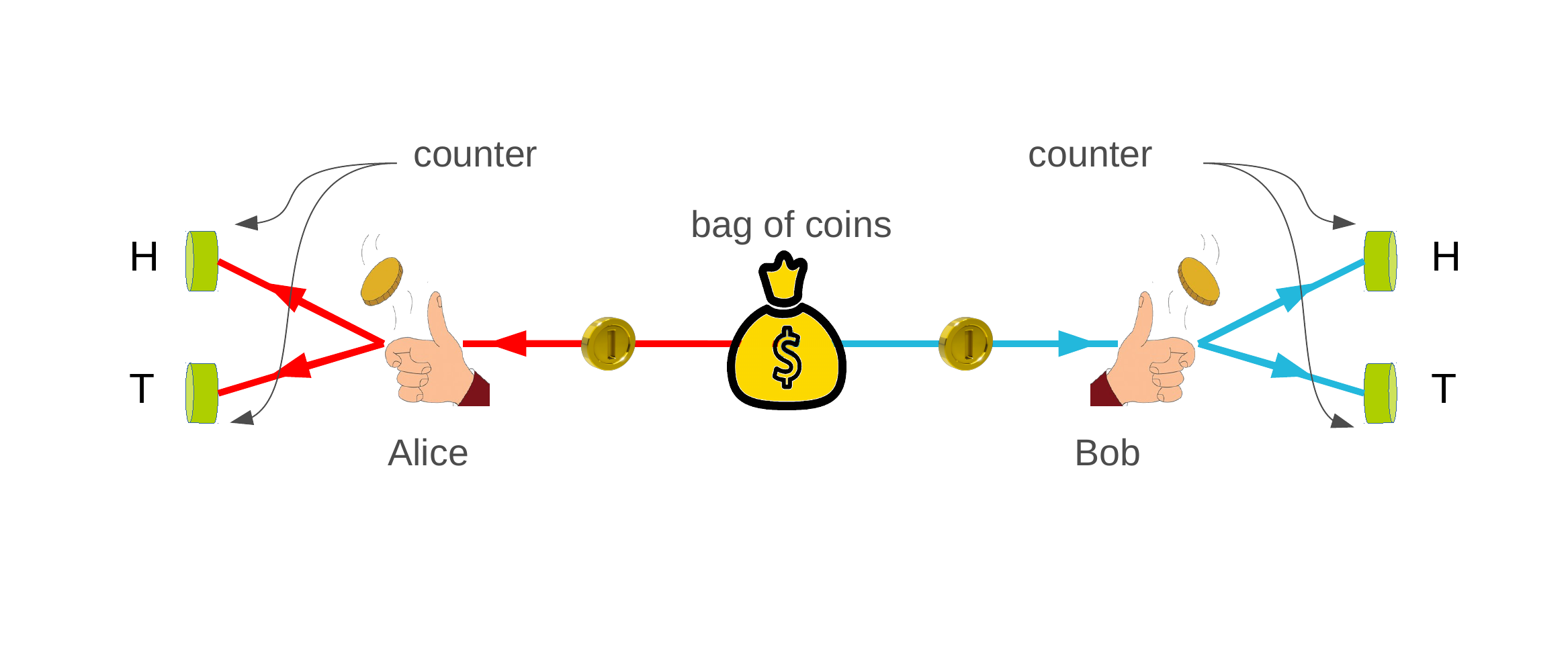}\vspace{-1.2cm}
 \caption{Coin toss experiment in Bernstein's world.}
  \label{fig:2}
\end{figure*}

\subsection{Entanglement }
We continue the previous example  and we set up a thought experiment
that uncovers the entanglement of the two coins. 

Assume two  coins are drawn  from a bag  in the  state \eqref{eq:state}.
We give the left coin to Alice and the right coin to Bob as depicted in Figure \ref{fig:2}.

We will now show that after the coins move apart, there are ``matching'' correlations between the output of their toss. That is a measurement (though a toss) of the bias of one coin
  will allow the prediction, with certainty, of the outcome of the measurement (toss) on the other coin.

Assume that Alice tosses her coin first and that it lands Heads, then she can  compute her  updated   prevision
(through \eqref{eq:regularPost}  with $\pi= \theta_1+\theta_3$) for
the gamble $q(\theta)=\theta_1+\theta_2$ (Heads on Bob's coin).
We can easily do this computation in the dual space.\footnote{Since the state \eqref{eq:state} is precisely specified,
we do not need to solve any optimisation  to compute the posterior prevision.} Note that
\begin{align}
\nonumber
0=L((q-\lambda_0)\pi)=-\lambda_0z_{001} - \lambda_0z_{100} + z_{011} \\
\nonumber
+z_{101} + z_{110} + z_{200}
\end{align} 
has solution  $\lambda_0=1$.  Alice instantaneously  knows that the result of the toss
of Bob's coin will be Heads.
Similarly, we can consider all the other cases
\begin{align}
\nonumber
q&=1-\theta_1-\theta_2,&\pi&=\theta_1+\theta_3,~ & \lambda_0&=0 \\
\nonumber
q&=\theta_1+\theta_2,&\pi&=1-\theta_1-\theta_3,~ &\lambda_0&=0 \\
\nonumber
q&=1-\theta_1-\theta_2,&\pi&=1-\theta_1-\theta_3,~  &\lambda_0&=1. 
\end{align}
This means that as soon as Alice sees the result of the  toss of the left coin, she immediately knows
that the result of the toss of Bob's  coin will be the same. The two coins  are totally ``correlated''.
Classical correlations can be explained by a common cause, or correlated ``elements of reality'' \citep{einstein1935can}. 
This is not the case in Bernstein's world.

In fact, although  the marginal operators satisfy
\begin{equation}
\begin{aligned}
 L(\theta_{H_r})&=L(\theta_1+\theta_2)&=z_{100}+z_{010}&=\frac{1}{2},\\
  L(\theta_{H_l})&=L(\theta_1+\theta_3)&=z_{100}+z_{001}&=\frac{1}{2},\\
 L(\theta_{H_r}^2)&=L((\theta_1+\theta_2)^2)&=z_{200}+2z_{110}+z_{020}&=\frac{1}{3},\\
 L(\theta_{H_l}^2):&=L((\theta_1+\theta_3)^2)&=z_{200}+2z_{101}+z_{002}&=\frac{1}{3},\\
\end{aligned} 
\end{equation}
these are the same moments we would get if the marginal distribution of the two coins is uniform.
That means that if we send an ensemble of coins to Alice (Bob), and she (he) tosses them, she (he) will experience Heads half of the times.
A classical correlation model that is compatible with these marginal moments   is
given by this probabilistic mixture of atomic charges (Dirac's deltas):
$$
p(\theta)=\frac{1}{2}\delta_{\left\{\begin{bmatrix}
                      \frac{1}{6} (3 - \sqrt{3})\\
                      0\\
                      0\\
                      \frac{1}{6} (3 + \sqrt{3})
                     \end{bmatrix}\right\}}(\theta)
                     +\frac{1}{2}\delta_{\left\{\begin{bmatrix}
                      \frac{1}{6} (3 + \sqrt{3})\\
                      0\\
                      0\\
                      \frac{1}{6} (3 - \sqrt{3})
                     \end{bmatrix}\right\}}(\theta).
$$
However, this probabilistic model (or any other) can never satisfy 
the moment constraints \eqref{eq:state} or, equivalently, can never violate the Bell's type inequality presented in the previous section.
We have entanglement!

\section{Conclusions}
In this work, after a brief description and analysis of the structural properties of P-coherent models, we have shown that the space of Bernstein polynomials in which nonnegativity is specified by the Krivine-Vasilescu  certificate is yet another instance of this theory and that therefore it is possible to construct in it a thought experiment uncovering entanglement with classical (hence non quantum) coins.


As a final side remark, we believe that formulating the theory of desirable gambles directly as a logic system  provides an elegant way for extending the framework to the accept-reject one but also for merging the latter with the theory of choice functions.

\bibliographystyle{plain}
\bibliography{biblio}
 
 \clearpage
\newpage

 \appendix
 
 \section{Technicalities}
 \subsection{Proofs of Section \ref{sec:tdg}}

\begin{proofof}{Theorem~\ref{prop:cons}}
 The equivalence $\cl\posi(\assess)=\closure_{\theory}(\assess)$ holds since in the derivation tree of a sequent $\assess \sequent g$, applications of the closure rule can be lifted up to the root and joint in a single inference step.
  The other equivalence traces back to \cite{walley91}.
\end{proofof}

 \subsection{Proofs of Section \ref{sec:con}}
 
 \begin{proofof}{Proposition~\ref{prop:cohe}}
   As for $\theory$, the equivalence $\cl\posi(\assess)=\closure_{\btheory}(\assess)$ holds since in the derivation tree of a sequent $\assess \sequent g$, applications of the closure rule can be lifted up to the root and joint in a single inference step.
  
Next, for the equivalence between items (1)-(3), first of all, notice that $\posi ( \assess \cup \bnonnegative) \subseteq \cl \posi ( \assess \cup \bnonnegative) \subseteq \closure_{\btheory}(\assess)$. Hence, (3) implies (2) implies (1). For the remaining implications we reason as follows. Assume that (2) holds, and assume $f + \delta \in \posi (\assess \cup \bnonnegative)$, for every $\delta > 0$. This means $f \in \cl(\posi (\assess \cup \bnonnegative))$.  Suppose $f \in \bnegative$, then since $\bnegative$ is open, $f + \delta \in \bnegative$ for some $\delta > 0$,
contradicting P-coherence of $\posi ( \assess \cup \bnonnegative)$. We therefore conclude that $f \notin \bnegative$, and that (3) holds. 
Now, assume (2) does not hold, i.e. $f \in \bnegative$ and $f \in  \posi ( \assess \cup \bnonnegative)$. Hence, $-f$ is in the interior of $\bnonnegative$, meaning that for some $\delta > 0$, $-f - \delta = g \in \bnonnegative$. From this we get that $-1 = \frac{g + f}{\delta} \in  \posi ( \assess \cup \bnonnegative)$: (1) does not hold.
\end{proofof}

 \begin{proofof}{Proposition~\ref{prop:eqq}}
 Since (1) always implies (2), we need to verify the other direction.
 Assume $\closure_{\btheory}(\assess)$ is not P-coherent. By Proposition \ref{prop:cohe}, $-1 \in \closure_{\btheory}(\assess)$, and thus $- \epsilon \in \closure_{\btheory}(\assess)$, for every $\epsilon \geq 0$. Let $f \in \gambles_R$. If (*) holds there is $\epsilon > 0$ such that $f+\epsilon \in \bnonnegative \subseteq \closure_{\btheory}(\assess)$.  Hence, by closure under linear combinations,  $f + \epsilon + (- \epsilon) = f \in  \closure_{\btheory}(\assess)$.
 
 Assume (b) holds. It is enough to prove that $-B \subseteq \posi(B \cup \{-1\})$. Fix $b \in B$. By hypothesis $-\frac{1}{\epsilon}b - \sum_{i =1}^\ell \lambda_i b_i = -1$, with $\lambda_i\geq 0$ and $\epsilon > 0$, for some $\{b_1, \dots, b_\ell\} \subseteq B$. Hence $-b = \epsilon(-1 + \sum_{i=1}^\ell \lambda_i b_i) \in \posi(B \cup \{-1\})$.
  We thus conclude that $\closure_{\btheory}(\assess) $ is logical inconsistent.
 \end{proofof}
 
 In the next Proposition we explicit another property of $\bnonnegative$ and verify that implies (*).
 
 \begin{proposition}
 \label{prop:basis}
 Assume that $\bnonnegative$ contains a basis $B$ of $\gambles_R$ and for every $b \in B$ there is a finite $\{b_1, \dots, b_\ell\} \subset B$ such that $b + \sum_{i=1}^\ell \lambda_i b_i = \epsilon > 0$, with $\lambda_i\geq 0$.
 Then condition (*) holds
 \end{proposition}
 \begin{proof}
 It is enough to check that, for $b \in B$, there is $\epsilon > 0$ such that $-b + \epsilon \in \bnonnegative$.
 But this is immediate since by hypothesis we know there is a finite $\{b_1, \dots, b_\ell\} \subset B$ such that $b + \sum_{i=1}^\ell \lambda_i b_i = \epsilon > 0$, with $\lambda_i\geq 0$. Hence $-b + \epsilon= \sum_{i=1}^\ell \lambda_i b_i$, which is in $\bnonnegative$ since the latter is a cone that includes $B$.
 \end{proof}

    \begin{proofof}{Proposition~\ref{eq:proppullBern}}
This follows by Proposition \ref{prop:basis} and the fact that Bernstein's polynomials form a partition of unity.
  \end{proofof}
 
   \begin{proofof}{Proposition~\ref{prop:coneBern}}
Assume that a polynomial $f(\theta)$ belongs to  \eqref{eq:b0simpl0}.
If there exists a monomial of $f(\theta)$ of degree $\ell$ less than $d$, then we can multiply
it for the Bernstein partition of unity of degree $d-\ell$. The resulting polynomial
will then belong to \eqref{eq:b0simpl}.
The opposite direction of the proof is obvious.
  \end{proofof}

\end{document}